\documentclass[11pt]{article}
\usepackage{hyperref}
\usepackage{amsfonts}
\usepackage[english]{babel}
\usepackage{a4wide,times}

\usepackage[latin1]{inputenc}
\usepackage[T1]{fontenc}
\usepackage{amssymb}
\usepackage{rotating, graphicx}
\usepackage{epsfig}
\usepackage{amsbsy}
\usepackage{verbatim}
\usepackage{color}
\usepackage{mathrsfs}
\usepackage{makeidx}
\usepackage{amsmath}
\usepackage{amsthm}
\usepackage{diagbox}
\usepackage{dsfont}	
\usepackage{float}
\usepackage{multirow}
\usepackage{subfigure}
\theoremstyle{plain}
\newtheorem{thm}{Theorem}[section]

\newtheorem{prop}[thm]{Proposition}

\newcommand{\argmin}{\arg\!\min}

\theoremstyle{definition}

\usepackage{dsfont}

\theoremstyle{remark}
\newtheorem{rem}{\bf Remark}[section]
\theoremstyle{remark}

\newtheorem{com*}{\bf Comment}

\usepackage{fancyhdr}

\DeclareMathOperator{\cF}{\mathcal{F}}
\DeclareMathOperator{\FF}{\mathbb{F}}
\makeatletter
\def \newequation#1#2{
   \@definecounter{#1}
   \@namedef{the#1}{\hbox{#2}}
   \@namedef{#1}{$$\refstepcounter{#1}}
   \@namedef{end#1}{
      \eqno \csname the#1\endcsname $$\global\@ignoretrue
      }
}
\makeatother
\newequation{E1}{($E_{b,\sigma,W}$)}
   
\makeatletter
\def \newequation#1#2{
   \@definecounter{#1}
   \@namedef{the#1}{\hbox{#2}}
   \@namedef{#1}{$$\refstepcounter{#1}}
   \@namedef{end#1}{
      \eqno \csname the#1\endcsname $$\global\@ignoretrue
      }
   }
\makeatother
\newequation{hyp3}{($\mathcal{H}_{b,\sigma}$)}

\makeatletter
\def \newequation#1#2{
   \@definecounter{#1}
   \@namedef{the#1}{\hbox{#2}}
   \@namedef{#1}{$$\refstepcounter{#1}}
   \@namedef{end#1}{
      \eqno \csname the#1\endcsname $$\global\@ignoretrue
      }
   }
\makeatother
\newequation{shleg}{(ShLeg)}

\makeatletter
\def \newequation#1#2{
   \@definecounter{#1}
   \@namedef{the#1}{\hbox{#2}}
   \@namedef{#1}{$$\refstepcounter{#1}}
   \@namedef{end#1}{
      \eqno \csname the#1\endcsname $$\global\@ignoretrue
      }
   }
\makeatother
\newequation{kl}{(KL)}

\makeatletter
\def \newequation#1#2{
   \@definecounter{#1}
   \@namedef{the#1}{\hbox{#2}}
   \@namedef{#1}{$$\refstepcounter{#1}}
   \@namedef{end#1}{
      \eqno \csname the#1\endcsname $$\global\@ignoretrue
      }
   }
\makeatother
\newequation{haar}{(Haar)}

\title{Conditional survival probabilities under partial information:\\ a recursive quantization approach with applications}
\author{  
  { Cheikh MBAYE} \thanks{Louvain Finance Center, UCLouvain, Voie du Roman Pays 34, 1348 Louvain-la-Neuve, Belgium, e-mail: {\tt cheikh.mbaye@uclouvain.be}. The research of C. Mbaye is funded by the National Bank of Belgium and an FSR grant.} \ \ \ {   \quad   Abass SAGNA }  \thanks{F\'ed\'eration de Math\'ematiques d'Evry,  Laboratoire Analyse et  Probabilit\'es, 23 Boulevard de France, 91037 Evry,  \& ENSIIE, e-mail: {\tt abass.sagna@ensiie.fr}.  This research  benefited from the support of the `` Chaire March\'es en Mutation'', F\'ed\'eration Bancaire Fran\c{c}aise.} \ \ \ \quad   { Fr\'ed\'eric VRINS} \thanks{Louvain Finance Center, UCLouvain, Voie du Roman Pays 34, 1348 Louvain-la-Neuve, Belgium, e-mail: {\tt frederic.vrins@uclouvain.be}.} 
}

\date{}

\begin{document}

\maketitle 

\begin{abstract}
We consider a structural model where the survival/default state is observed together with a noisy version of the firm value process. This assumption makes the model more realistic than most of the existing alternatives, but triggers important challenges related to the computation of conditional default probabilities. In order to deal with general diffusions as firm value process, we derive a numerical procedure based on the recursive quantization method to approximate  it. Then, we investigate the error approximation induced by our procedure. Eventually, numerical tests are performed to evaluate the performance of the method, and an application is proposed to the pricing of CDS options.
\end{abstract}
\medskip


\section{Introduction}
In the recent decades, credit risk received an increasing attention from academics and practitioners. 
In particular, the 2008 financial crisis shed the light on the importance of having sound credit risk models to better asses the default likelihood of firms and counterparties. The structural approach is one of the two most popular frameworks. It is originated to the seminal work of Merton \cite{Mer} and uses the dynamics of structural variables of a firm, such as asset and debt, to determine whether the firm defaulted before a given maturity. 
To better deal with the actual timing of the default event, first passage time models were then introduced. Among them is the celebrated Black and Cox model \cite{BlaCox} which adds a time-dependent barrier, among others. 
Yet, the Black and Cox model has few parameters and is not easily calibrated to structural data such as CDS quotes along different maturities. To that end, extensions of the same models called AT1P and SBTV were introduced in \cite{BrigTar04} and \cite{BrigTar05} allowing exact calibration to credit spreads using efficient closed-form formulas for default probabilities. 
\medskip

In practice however, it is difficult for investors to perfectly assess the value of the firm's assets. 
In this case, modeling the firm value in a Black-Cox framework is problematic, since the model assumes that the firm's underlying assets are observable. Moreover, in such a framework, the default time is predictable, leading to vanishing credit spreads for short maturities. In order to address these drawbacks, Duffie and Lando \cite{DufLan} proposed a model where the investors have only partial information on the firm value and observe at discrete time intervals a noisy accounting report. The default time becomes totally inaccessible in the market filtration. 
As a result, the corresponding short term spreads are always higher compared to the complete information short term spreads. 
Alternatively, some extensions of this model based on noisy information in continuous time can be found among others in \cite{CocGemJea} or recently in \cite{Fre19}.\medskip 

In this paper, we both generalize and improve results derived in earlier studies. The models presented in \cite{CalSag} and \cite{ProSag14} can deal with arbitrary firm-value diffusions, but are heavy. Moreover, the considered information flow is only made of a noisy version of the firm-value. This is not realistic as in practice, investors can obviously observe the default state of the firm. This larger filtration is considered in \cite{CocGemJea}, but under the restrictive assumption that the firm-value process is a continuous and invertible function of a Gaussian martingale. In this work, we consider the same information set as \cite{CocGemJea} but relax the restriction regarding the firm-value dynamics. 
To deal with this general case, we propose a numerical scheme based on fast quantization recently introduced in \cite{PagSagMQ}. This technique is faster compared to \cite{CalSag} and \cite{ProSag14} as there is no need to rely on Monte-Carlo simulations to compute the conditional survival probabilities. 
A detailed analysis of the error induced by the approximation is provided. Eventually, we illustrated our method on the pricing of CDS option credit derivatives.\medskip

The reminder of the paper is organized as follows. In Section \ref{sec:mdel}, we introduce the model. Different information flows and corresponding survival probabilities will be discussed. Section \ref{sec:Quanti} presents the estimations of the survival probabilities using the recursive quantization method and the stochastic filtering theory. We then give a brief introduction to the quantization method before deriving the error analysis pf these estimations. Section \ref{sec:numerics} is devoted to the results of the numerical experiments.

\section{The model}\label{sec:mdel}
Assume a probability space $(\Omega, \mathcal F,\mathbb P)$, modeling the uncertainty of our economy.  We consider a structural default model, and represent the default time $\tau_X$ of a reference entity as the first passage time of firm value process,below a default threshold. More precisely, we consider a Black and Cox setup \cite{BlaCox}, where the stochastic process $X$ represents the actual value of the firm and $\bf{a} \in \mathbb R$ stands for  the default barrier. Assuming $\tau_X>0$, we have:
\begin{equation}\label{eq:defTau}
\tau_{X} := \inf \left\{u \geq 0 : X_u \le {\bf{a}} \right\},~~0< {\bf{a}} < X_0 
\end{equation}
where $\inf \emptyset := + \infty$, as usual. We restrict ourselves to consider $0\leq t\leq T$ where $T$ is a finite time horizon.  

We consider a partial information model where the true firm value $X$ (called \emph{signal process} hereafter) is not observable and we only observe $Y$ (\emph{observation process}), which is correlated with $X$. We suppose that    the dynamics of $X$ and $Y$ are governed by the following stochastic differential equations (SDEs) :
\begin{equation} \label{EqSignalStatePr}
 \begin{cases}
 dX_t = b(t,X_t) dt + \sigma(t,X_t) dW_t,  &   X_0=x_0, \\
 dY_t = h(t,Y_t,X_t) dt + \nu(t,Y_t) dW_t +  \delta(t,Y_t) d\widetilde{W}_t, &  Y_0=y_0,
 \end{cases} 
 \end{equation}
\noindent
where $(W,\widetilde W$) is a standard two-dimensional Brownian motion. We suppose that the functions $b, \sigma,\nu,\delta:[0,+ \infty) \times \mathbb R \rightarrow \mathbb R$ are Lipschitz in $x$ uniformly in $t$ and  that $\sigma(t,x)>0$ for every $(t,x) \in [0, + \infty) \times \mathbb{R} $ . These conditions ensure that the above SDEs admit a unique strong solution. Moreover we assume that $h$ is locally bounded and Lipschitz in $(y,x)$, uniformly in $t$ and that  $\nu(t,y)>0$ and  $\sigma(t,y)>0$ for every $(t,y) \in [0, + \infty) \times \mathbb{R}$.


\subsection{Information flows}
One of the major critiques of such models is that in practice, the firm value is not observable. It is therefore not realistic to consider that the information available to the investor is $\FF^X:=(\cF^X_t)_{t\geq 0}$,  $\cF^X_t:=\sigma(X_u,0\leq u\leq t)$. A more realistic framework has been proposed in \cite{CalSag} where the investor information is given by the natural filtration $\FF^Y$ of  a noisy version $Y$ of the process $X$. However, one might argue that this way of modeling the information is not realistic either since given $\cF^Y_t$, the investor is unable to know whether the reference entity defaulted or not by time $t$. In other words, the default indicator process $H=(H_t)_{t\geq 0}$, $H_t := \mathds{1}_{ \{t \ge \tau_X  \}}, t \ge 0$, is not adapted to $\FF^Y$.\\

In this paper, we address this point by considering a more realistic information flow, defined as the progressive enlargement of $\FF^Y$ with $\FF^H$, the natural filtration of the default indicator process, $$
\cF_t^H:= \sigma (H_u, 0 \le u \le t), \quad t \ge 0.
$$
In other words, we have two investors' information flows.  
On the one hand we have the information available to the common investor, defined as the progressive enlargement of the natural filtration of the default indicator process with that of the noisy firm-value, noted $\FF=(\cF_t)_{t\geq 0}$ where
$$
\cF_t := \cF_t^Y \vee \cF_t^H, \quad t \ge 0\;.
$$
In this setup,  the following relationships hold:
$$\FF^H\subsetneq \FF^X=\FF^W \subsetneq \mathbb G~~~~ \text{ and }~~~~ \FF^Y \subsetneq \mathbb G\;.$$
where $\mathbb G:=(\mathcal{G}_t)_{t\geq 0}$  is the full information, i.e., the information available for example to a small number of stock holders of the company, who have access to $Y$ and $X$.\medskip

 On the other hand, the natural filtration of the actual (i.e. noise-free) firm-value process, $\FF^X$, could be seen as the information available to insiders. Economically, $X$ would represent the value of the firm, which is unobservable to the common investors, while $Y$ might be the market price of an asset issued by the firm, accessible to all market participants, and $\bf{a}$ would stand for the solvency capital requirement imposed by regulators. 

\begin{rem}
The results in the paper can be straightforwardly extended in the case when the default barrier ${\bf{a}}$ is a piecewise constant function of time ${\bf{a}}: [0,\infty) \rightarrow [0, \infty)$, with $0 < {\bf{a}}(0) < x_0$. 
The other extensions are beyong the scope of this paper. Nevertheless, let us mention that in, e.g., \cite{PotzWang} crossing probabilities for the Brownian motion are obtained in the case when the (double) barrier is a piecewise linear function on $[0,T]$ and approximations for crossing probabilities are obtained for general nonlinear bounds. 
\end{rem}



\noindent


\subsection{Survival probability}
A fundamental output of a credit model is the survival probability of the firm up to time $t$ conditional upon $\cF_s$, $s\leq t\leq T$: 
\begin{equation}\label{eq:survivalProb}
\mathbb P\left( \tau_X > t  \vert \mathcal F_s\right)= \mathbb E \left( \mathds{1}_{ \{\tau_X > t \}} \Big\vert \mathcal F_s\right)
\end{equation}
Note that this probability collapse to zero whenever $\{ \tau_X \leq s \}$.  Recall that the specificity of our approach is that, the actual value of the firm $X$ is not revealed in $\mathcal{F}_s$; only a noisy version $Y$ is accessible. 

 Using the Markov property of $X$, the fact that the two Brownian motions are independent and the chain rule of the conditional expectation, we show  the following result.\\
 \begin{prop}  We have, for $s \le t$,
 \begin{equation}  \label{eq:start}
 \mathbb P\left( \tau_X > t  \Big\vert \mathcal F_s\right)
 = \mathds{1}_{ \{ \tau_X > s  \}} \frac{ \mathbb E \left[ \mathds{1}_{ \{ \tau_X > s  \}} F(s,t,X_s) \vert \mathcal F_s^{Y} \right] }{\mathbb P \left(  \tau_X > s    \Big\vert \mathcal F_s^{Y} \right)}
 \end{equation}
 where, for every $x \in \mathbb{R}$,
\begin{equation} \label{DefinitionDeF}
 F(s,t,x) := \mathbb P \left( \inf_{s < u \le t} X_u >{\bf{a}} \Big\vert X_s =x \right)
 \end{equation}
is the  conditional survival probability under full information. Furthermore, it holds on the set  $\{ \tau_X>s\}$ that 
\begin{equation}\label{eq:defaultProbComp}
\mathbb P\big(\tau_{ X} >  t \vert {\cal F}_s^{ Y}\big)  \le \mathbb P\big(\tau_{ X} > t \vert { \cF}_s\big).
\end{equation}
\end{prop}

\begin{proof} 
Using a key result in the theory of conditional expectations commonly referred to as the \emph{Key lemma} (see e.g. Lemma 3.1 in \cite{EllJeaYor}) we have
\begin{eqnarray}
   \mathbb P\left( \tau_X > t  \Big\vert \mathcal F_s\right)
   &=&\mathds{1}_{ \{ \tau_X > s  \}} \mathbb E \left( \mathds{1}_{ \{\tau_X > t \}} \Big\vert \mathcal F_s^{Y} \vee \mathcal F_s^H \right)\nonumber\\ &=&  \mathds{1}_{ \{ \tau_X > s  \}} \frac{\mathbb E \left( \mathds{1}_{ \{ \tau_X > t  \}}  \Big\vert \mathcal F_s^{Y} \right)}{\mathbb P \left(  \tau_X > s    \Big\vert \mathcal F_s^{Y} \right)} \nonumber  \\
 & =  & \mathds{1}_{ \{ \tau_X > s  \}} \frac{\mathbb P \left(  \tau_X > t  \vert \mathcal F_s^{Y} \right)}{\mathbb P \left(  \tau_X > s    \Big\vert \mathcal F_s^{Y} \right)}.  \label{EqInegBetweenDefaultProb}
 \end{eqnarray}
The proof is completed by noting that $\mathbb P \left( \tau_X > t    \vert \mathcal F_s^{Y} \right)=\mathbb E\left[\mathds{1}_{ \{ \tau_X > s  \}} F(s,t,X_s) \vert \mathcal F_s^{Y} \right]$ (see e.g. \cite{ProSag14}). 
The particular case follows from \eqref{EqInegBetweenDefaultProb} since on the event $\{ \tau_X>s \}$, 
$$\mathbb P\big(\tau_{ X}>t \vert {\cF}^{ Y}_s\big)=\mathbb P(\tau_{ X} >t \vert {\cF}_s)\,\mathbb P\big(\tau_{ X}>s \vert {\cal F}^{ Y}_s \big)\leq \mathbb P(\tau_{ X} >t \vert { \cF}_s) $$
or, equivalently, $\mathbb P\big(\tau_{ X} \le t \vert { \cF}_s\big) \le \mathbb P\big(\tau_{ X}\le t \vert {\cal F}_s^{ Y}\big).$
\end{proof}
\begin{rem}
The inequality \eqref{eq:defaultProbComp}, confirmed by the numerical experiments in Section \ref{sec:numerics}, means that the less we have  information on the state of the system ($\cF_s^{ Y} \subset \cF_s^{ Y} \vee \cF_s^{ H}={\cF}_s$), the higher the default probability. This also shows the difference with \cite{ProSag14} where the quantity of interest is just  $\mathbb P\big(\tau_{ X}>t \vert {\cal F}_s^{ Y}\big)$.
\end{rem}

\subsection{The problem}


Note that we have clearly stated the expression of interest, namely  the survival probability of the reference entity up to time $t$ conditional upon the investor's information up to time $s$, 
we need to actually compute it. In order to even more comply with real market practice, we further consider that we can only access to, say, \emph{discrete time} observations of $Y$ up to time $s$.   To that end, let us start by fixing a time discretization grid over $[0,t]$:
$$
0=t_0< \dots<t_m=s < t_{m+1}<\dots< t_n=t.
$$

Our aim  is to approximate the right hand side of \eqref{eq:start} by recursive quantization. In some specific models (those for which \eqref{EqSignalStatePr} admits an explicit solution $(X,Y)$, like in the Black-Scholes framework), we will consider the discrete trajectories $(X_{t_k}, Y_{t_k})_{k=0, \ldots, n}$. In more general models, we need to make a discrete time approximation of the quantity of interest. To this end, we suppose that we have access to a 
trajectory of $Y$ sampled at $m$ times: 
$(\bar Y_{t_0}, \ldots,\bar Y_{t_m})$,  with $t_0=0$ and $t_m = s$ (which in practice will be approximated  from the paths of the  Euler scheme  associated to the stochastic process $Y$) and will estimate \eqref{eq:survivalProb}
by
\[
 \frac{\mathbb P \left( \tau_X > t    \Big\vert {\cal F}_s^{\bar Y} \right)}{\mathbb P \left(  \tau_X > s    \Big\vert {\cal F}_s^{\bar Y} \right)}, 
\]
on the event $ \{ \tau_X > s  \}$, where  ${\cal F}_s^{\bar Y}$ $= \sigma(\bar Y_{t_k}, \ t_k \le s)$ $= \sigma(\bar Y_{t_0}, \ldots, \bar Y_{t_m})$.

\subsection{Discrete time approximation}
We denote by $\bar X$ the continuous Euler scheme associated  to the process $X$ in Equation \eqref{EqSignalStatePr}, namely:
$$
\bar{X}_s  = \bar{X}_{\underline{s}} + b( \underline{s},\bar{X}_{\underline{s}}) (s - \underline{s}) + \sigma( \underline{s},\bar{X}_{\underline{s}}) (W_{s} - W_{\underline{s}}), \quad  \bar{X}_0 = x_0,
$$
with $\underline{s} = t_k $ if $s \in [t_k,t_{k+1})$, for $k=0,\ldots,n$.
Based on the Euler scheme, we introduce the discretized version of our state-observation processes $(\bar X, \bar Y)$
\setlength\arraycolsep{1pt}
\begin{equation}  
\begin{cases} 
 \bar X_{t_{k+1}} = \bar{X}_{t_k} +  b(t_k,\bar{X}_{t_k}) \Delta_k + \sigma(t_k,\bar X_{t_k}) (W_{t_{k+1}}-W_{t_k})   \\
 \bar  Y_{t_{k+1}} =  \bar  Y_{t_k}  + h(t_k,  \bar  Y_{t_k}, \bar X_{t_k}) \Delta_k + \nu(t_k,  \bar Y_{t_k}) (W_{t_{k+1}}-W_{t_k}) + \delta(t_k,  \bar  Y_{t_k}) (\widetilde W_{t_{k+1}} -\widetilde W_{t_{k}})
 \end{cases}
 \end{equation} 
where  $k \in \{ 0, \dots, n-1 \}$ for the signal process and $k \in \{0,\dots,m-1\}$ for the observation process and where $\Delta_k := t_{k+1}-t_{k}$.

 Supposing that  we have access to a discrete trajectory of $Y$, $( \bar  Y_{t_0},\ldots, \bar  Y_{t_m})$, our first goal is  to approximate (recall that $t_m=s$)
\begin{equation} \label{EqDisCondHitTimeEstim}
\frac{\mathbb P \big(\tau_{ X} > t \vert\, \mathcal F_s^Y \big)}{\mathbb P \big(\tau_{X} > s \vert\, \mathcal F_s^Y \big)} \quad \textrm{by} \quad \frac{\mathbb P \big(\tau_{\bar X} > t \vert {\cal F}_s^{\bar Y}\big)}{\mathbb P \big(\tau_{\bar X} > s \vert {\cal F}_s^{\bar Y}\big)},
\end{equation} 
where (recall Equation (\ref{eq:defTau}))
$$ 
\tau_{\bar X} := \inf\{ u \geq 0,  \bar{X}_u \leq \bf{a}  \}.
$$

Using the Brownian Bridge method and the Markov property of  $( \bar X_{t_k}, \bar Y_{t_k})_{k}$,  we show that the quantity  (\ref{EqDisCondHitTimeEstim}) can be written in  a closed  formula.  


\begin{thm} \label{ThmMainResultPhiNLF}
We have:
\begin{equation}   \label{EqForFixedObservationNLF}
\frac{\mathbb P \big(\tau_{\bar X} > t \vert {\cal F}_s^{\bar Y}\big)}{\mathbb P \big(\tau_{\bar X} > s \vert {\cal F}_s^{\bar Y} \big)}  =  \Psi(\bar Y_{t_0},\dots,\bar Y_{t_m}),
\end{equation}
where for  $ y=(y_0,\dots,y_{m}) \in \mathbb R^{m+1}$,  
\begin{equation}  \label{EqExplicitCompSurvivalProba}
\Psi (y) = \frac{\mathbb E \big[\bar{F}(t_m,t_n,\bar X_{t_m}) K^{m}_{\bf{a}}   L_{ y}^m \big] }{ \mathbb E[K^{m}_{\bf{a}}  L_y^m]},
\end{equation}
with 
$$  K^{m}_{\bf{a}} =  \prod_{k=0}^{m-1}  G_{\Delta_k\sigma_k^2}^{\bar{X}_{t_k},\bar{X}_{t_{k+1}}}({\bf{a}}), \qquad  L_y^m =  \prod_{k=0}^{m-1} g_k(\bar X_{t_k},y_k;\bar X_{t_{k+1}},y_{k+1})$$
and where for every $x \in \mathbb R$,
\begin{equation}  \label{EqDefFuncBarF}
 \bar{F}(t_m,t_n,x)  =  \mathbb E \Big[ \prod_{k=m}^{n-1}   G_{\Delta_k\sigma_k^2}^{\bar{X}_{t_k},\bar{X}_{t_{k+1}}}({\bf{a}}) \big\vert\, \bar{X}_{t_m}=x  \Big].
 \end{equation}
The function $g_k$ 
is defined by   
\begin{eqnarray}
 g_k(x_k,y_k;x_{k+1},y_{k+1}) &=&  \frac{\mathbb P\big((\bar X_{t_{k+1}},\bar Y_{t_{k+1}}) = (x_{k+1},  y_{k+1}) \vert (\bar X_{t_{k}},\bar Y_{t_{k}}) = (x_{k},  y_{k}) \big)}{\mathbb P\big(\bar X_{t_{k+1}} = x_{k+1} \vert \bar X_{t_k} =x_k\big)} \nonumber \\
 &=&\frac{1}{(2\pi \Delta_k)^{1/2}   \delta_k  }  \exp\Bigg(  -     \frac{\nu_k^2}{2 \delta^2_k \Delta_k} \Big( \frac{x_{k+1}  - m_k^1}{\sigma_k } -   \frac{y_{k+1} -   m_k^2}{\nu_k}  \Big)^2 \Bigg) \label{EqDefFunctg_k}
\end{eqnarray}

\noindent with  $m_k^1 := x_k + b_k \Delta_k\;$ and $\;m_k^2 := y_k + h_k  \Delta_k$. Finally,  
\begin{eqnarray} \label{EqDefOfG}
   G_{\Delta_k\sigma_k^2}^{x_{k},x_{k+1}}({\bf{a}}) &=& \mathbb P\big(\inf_{u \in [t_k,t_{k+1}]} \bar X_u \ge {\bf a} \vert \bar X_{t_k} =x_k\big)\nonumber\\
   &=& \left( 1-\exp\left(-  \frac{ 2 (x_k- {\bf{a}}) (x_{k+1}- {\bf{a}})}{ \Delta_k \sigma^2  (t_k,x_k) }\right) \right) \mathds {1}_{\{x_k \geq {\bf{a}};\; x_{k+1} \geq   {\bf{a}}\} }.  
   \end{eqnarray}

\end{thm}

 \begin{proof}
 Following Theorem 2.5. in \cite{ProSag14}, we have, 
  \[
   \mathbb P \big(\tau_{\bar X} > t \vert {\cal F}_s^{\bar Y} \big) = \frac{\mathbb E \big[\bar{F}(t_m,t_n,\bar X_{t_m}) K^{m}_{{\bf{a}}}  L_{ y}^m \big] }{ \mathbb E[   L_y^m]}  
      ~~\text{ and }~~ 
     \mathbb P \big(\tau_{\bar X} > s \vert {\cal F}_s^{\bar Y} \big)= \frac{\mathbb E \big[ K^{m}_{{\bf{a}}}  L_{ y}^m \big] }{ \mathbb E[  L_y^m]}.
     \]
 \end{proof}

The question of interest is now to know how to estimate efficiently  $\Psi(y)$ for $y=(\bar Y_{t_0}, \ldots, \bar Y_{t_m})$.  Owing to the form of the random vector $ K_{{\bf{a}}}^{m}$, we may put it together with $ L_y^m$ to be reduced  to similar formula as the filter estimate in a standard nonlinear filtering problem. In other work we may write  for  $ y:=(y_0,\dots,y_{m})$,  

\begin{equation*}  
\Psi (y) = \frac{\mathbb E \big[\bar{F}(t_m,t_n,\bar X_{t_m}) \,  L_{ y,\bf{a}}^m \big] }{ \mathbb E [L_{y,{\bf{a}}}^m]} 
~~\text{ where }~~ L_{y,{\bf{a}}}^m =  \prod_{k=0}^{m-1}  g_k(\bar X_{t_k},y_k;\bar X_{t_{k+1}},y_{k+1}) {\small \times} G_{\Delta_k\sigma_k^2}^{\bar{X}_{t_k},\bar{X}_{t_{k+1}}}({\bf{a}})\;.
\end{equation*}

\section{Approximation  by recursive quantization}\label{sec:Quanti}
Notice that defining the operator  $\pi_{y,m}$, for every bounded measurable function $f$, by 
$$  \pi_{y,m} f := \mathbb E \big[f( \bar X_{t_m}) L_{y, {\bf{a}}}^m  \big],$$
we  have
\begin{equation} \label{EqRedefinePi}
 \Psi(y)   = \frac{\pi_{y,m} \bar F(t_{m},t_n,\cdot)}{\pi_{y,m} \mbox{\bf 1}}=:  \Pi_{y,m} \bar F(t_m,t_n,\cdot),   
 \end{equation}
where $\mbox{\bf 1}(x) =1$, for every real $x$.  Then  it is enough  to tell how to compute the numerator 
\begin{equation*}
 \pi_{y,m} \bar F(t_{m},t_n,\cdot)  =  \mathbb E \Big[ \bar F(t_{m},t_n,\bar X_{t_m}) \prod_{k=0}^{m-1} g_k^{\bf{a}}( \bar X_{t_k},y_k; \bar X_{t_{k+1}}, y_{k+1}) \Big]
 \end{equation*}
where 
 $$ g_k^{\bf{a}}( \bar X_{t_k},y_k; \bar X_{t_{k+1}}, y_{k+1})  =  g_k( \bar X_{t_k},y_k; \bar X_{t_{k+1}}, y_{k+1}) \times G_{\Delta_k \sigma_k^2}^{\bar X_{t_{k}},\bar X_{t_{k+1}}} ({\bf{a}}).$$
 
At this stage, several methods involving Monte Carlo simulations as the particle method can be used   to approximate  $\Pi_{y,m}$.   Optimal quantization is an alternative and some times as  a substitute  to the Monte Carlo method  to  approximate such a quantity (we  refer to  \cite{Sel} for a comparison of  particle like methods and optimal quantization methods). 
 
 To use the  optimal quantization methods we  have to quantize   the marginals  of the process $(\bar X_{t_k})_k$, means, to represent every marginal $\bar X_{t_k}$, $k=0,\ldots,n$,  by a discrete random variable  $\widehat X_{t_k}^{\Gamma_{k}}$ (we will simply denote it  $\widehat X_{t_k}$ when there is no ambiguity) taking $N_k$ values $\Gamma_k = \{x^k_1,\ldots,x^k_{N_k}  \}$. As we will see later,  we have also need in our context  to compute the transition probabilities $\hat p_k^{ij} = \mathbb P (\widehat X_{t_{k}} = x^{k}_j \vert \widehat X_{t_{k-1}} = x^{k-1}_i)$, for $i=1,\ldots,N_{k-1}$; $j=1,\ldots,N_k$.  To this end,  we may use stochastic algorithms to get the optimal grids and the associated transition probabilities (see e.g. \cite{PagPha, Sel}).     This method  works well but may be  very time consuming.  The so-called marginal functional quantization method (see \cite{CalSag, ProSag14, Sag}) is used as an alternative to the  previous method.  It  consists  to construct the marginal quantizations  by  considering the ordinary differential equation (ODE)  resulting to  the substitution of  the Brownian motion appearing in the dynamics of $X$ in \eqref{EqSignalStatePr}  by a quadratic quantization  of  the  Brownian motion (see \cite{LusPag1}).   This procedure performs the marginal quantizations  quite instantaneous and  works well enough  from the numerical point of view even if the rate of convergence (which has not been computed yet from the theoretical point of view) seems to be poor.   As an alternative to the  two previous  methods, we propose the recursive   marginal  quantization (also called  fast quantization) method introduced in \cite{PagSagMQ}.  It consists of quantizing the process  $(\Bar X_{t_k})_{k=0,\ldots,n}$, based on a recursive method involving  the conditional distributions   $\Bar X_{t_{k+1}} \vert \Bar X_{t_k}$, $k=0,\ldots,n-1$. For the problem of interest, this last method is more performing than the previous ones due to its computation speed  and to its robustness.  
 
  On the other hand, the function $\bar F$ has been estimated by Monte Carlo method in \cite{CalSag, ProSag14}. For competitiveness reasons of   the recursive quantization w.r.t. the previously raised methods, we propose here  to approximated both quantities $\Pi$ and $\bar F$ by the recursive quantization method. 

\subsection{Approximation of $\Pi_{y,m}$ by recursive quantization} 

 Given that the denominator  in the right hand side of   \eqref{EqRedefinePi} has a similar form as the  numerator, we will only show how to compute the numerator.  We remark that   $\pi_{y,m}$ can be computed from the following recursive formula:

\begin{equation} \label{Eqpi_n fIntro}
\pi_{y,k}   = \pi_{y,k-1} H_{y,k},  \qquad k=1,\dots, m,
\end{equation}
where, for every $k=1,\ldots,m$, and for every  bounded and measurable function $f$,  the transition kernel $H_{y,k}$  is defined by
 %
 \begin{equation*} \label{def_kernel_H0}
 H_{y,k} f(z)   =   \mathds{E} \left[ f(\bar{X}_{t_k})  g_{k-1}^{\bf{a}}(\bar X_{t_{k-1}},y_{k-1};  \bar X_{t_{k}}, y_{k})  \vert  \bar{X}_{t_{k-1}} =z \right] 
~~\text{ with }~~
 H_{y,0} f:= \mathbb{E}[f(\bar{X}_0)]\;.
 \end{equation*}
%

 In fact, for any bounded Borel function $f$ we have
 \begin{eqnarray*}
 \pi_{y,k} f  &=& \mathbb E\Big[ f(\bar X_{t_k}) \prod_{\ell=0}^{k-1}g_{\ell}^{\bf a}(\bar X_{t_{\ell}},y_{\ell};\bar X_{\ell+1},y_{\ell+1}) \Big] \\
 &=& \mathbb E\Big[ \mathbb E\Big(f(\bar X_{t_k}) \prod_{\ell=0}^{k-1}g_{\ell}^{\bf a}(\bar X_{t_{\ell}},y_{\ell};\bar X_{\ell+1},y_{\ell+1}) \big \vert {\cal F}_{k-1}^{\bar X} \big)\Big].
  \end{eqnarray*}
 Since $\bar X$ still be a Markov process we deduce that 
 \begin{eqnarray*}
 \pi_{y,k}f &=& \mathbb E\Big[ \mathbb E\Big(f(\bar X_{t_k}) g_{k-1}^{\bf a}(\bar X_{t_{k-1}},y_{k-1};\bar X_{k},y_{k}) \big \vert {\cal F}_{k-1}^{\bar X} \big) \prod_{\ell=0}^{k-2}g_{\ell}^{\bf a}(\bar X_{t_{\ell}},y_{\ell};\bar X_{\ell+1},y_{\ell+1}) \Big] \\
 & = & \mathbb E\Big[ H_{y,k}f(\bar X_{t_{k-1}}) \prod_{\ell=0}^{k-2}g_{\ell}^{\bf a}(\bar X_{t_{\ell}},y_{\ell};\bar X_{\ell+1},y_{\ell+1}) \Big] \\
 & = & \pi_{y,k-1} H_{y,k}f.
 \end{eqnarray*}
 
 Then, when we have access to the quantization of the marginals of the process $\bar X$, the functional $\pi_{y,k}$ can be approximated recursively by optimal quantization as $ \hat {\pi}_{y,k} = \hat{\pi}_{y,k-1} \widehat H_{y,k}$ where for every $k \ge 1$, $\widehat H_{y,k}$ is a matrix $N_k \times N_{k-1}$ which components $\widehat{H}_{y,k}^{i,j}$ read 
 \[
\widehat{H}_{y,k}^{ij} =  g_{k-1}^{\bf{a}}(x^i_{k-1},y_{k-1};  x^j_k, y_{k})  \,  \hat{p}_k^{ij}\, \delta_{x_k^j} 
 \]
 where 
 \[
 p_{k}^{ij} = \mathbb P(\widehat{X}_{t_k} = x^j_{k} \vert \widehat{X}_{t_{k-1}} = x^i_{k-1})
 \]
and $(\widehat{X}_{t_k})_{k}$ is the quantization of the process $(\bar X_t)_{t \geq 0}$  over the time steps $t_k, k=1,\dots,m$:    on the grids  $\Gamma_k = \{ x^1_k,\dots,x^{N_k}_k \}$, of sizes $N_k$.

As a consequence, the quantity of interest $\Pi_{y,m} \bar F(t_m,t_n,\cdot) $ is estimated by  
\begin{equation}  \label{EqPiF(t_n,t_m,)}
 \widehat{\Pi}_{y,m} \bar F(t_m,t_n,\cdot)  = \sum_{i=1}^{N_m}\widehat{\Pi}^i_{y,m} \bar  F(t_m,t_n,x_m^i).  
 \end{equation}
where
$$  \widehat{\Pi}^i_{y,m}  := \frac{\widehat {\pi}_{y,m}^i}{\sum_{j=1}^{N_m}  \widehat{\pi}_{y,m}^{i}}, \quad i=1,\dots,N_m $$
and where $\widehat {\pi}_{y,m}$ is  the estimation (by optimal quantization) of $\pi_{y,m}$  defined recursively by

\begin{equation}\label{eq:hatsmallpi}
\left \{
\begin{array}{ll}
\widehat{\pi}_{y,0} = \widehat H_{y,0} \\
\widehat{\pi}_{y,k} = \widehat{\pi}_{y,k-1} \widehat H_{y,k}:=\Big[   \sum_{i=1}^{N_{k-1}} \widehat H_{y,k}^{i,j} \widehat{\pi}_{y,k-1}^{i}  \Big]_{j=1,\dots,N_k},\quad  k=1,\dots,m 
\end{array}
\right.
\end{equation}
with  
\begin{equation}
\widehat{H}_{y,k}^{ij} =  g_{k-1}^{\bf{a}}(x^i_{k-1},y_{k-1};  x^j_k, y_{k})  \,  \hat{p}_k^{ij}\, \delta_{x_k^j} .  
\end{equation}



Our aim is now  to use the (marginal) recursive quantization  method to estimate the  $F(t_m,t_n,x_m^i)'$s. 

\subsection{Approximation of $\bar F(t_m,t_n,\cdot)$ by recursive quantization}

Recall that for every $x$,
\[
 \bar  F(t_m,t_n,x) =   \mathbb E \Big( \prod_{k=m}^{n-1}  G_{\Delta_k\sigma_k^2}^{\bar X_{t_{k}}, \bar X_{t_{k+1}}}({\bf{a}})  \big\vert\, \bar{X}_{t_m} =x \Big). 
 \]
 As previously, we remark that if we define  the functional $\pi_{n,m}$ by
$$\big(\pi_{n,m} f\big)(x)=  \mathbb E \Big(f(\bar{X}_{t_n})  \prod_{k=m}^{n-1} G_{\Delta_k\sigma_k^2}^{\bar X_{t_{k}}, \bar X_{t_{k+1}}}({\bf{a}}) \big\vert\, \bar{X}_{t_m} =x \Big), $$
for every bounded and measurable function  $f$, then $F(t_m,t_n,x)$ reads
$$ \bar F(t_m,t_n,x)  = \big(\pi_{n,m} \mbox{\bf{1}}\big)(x). $$
Now,  for every  bounded and measurable function $f$, defining  as previously   the transition kernel $H_{k}$ as, 
 \begin{eqnarray*}   
 \big(H_{k} f\big)(z)  &  = &  \mathbb{E} \left( f(\bar{X}_{t_k})  G_{\Delta_{k-1}\sigma_{k-1}^2}^{\bar X_{t_{k-1}}, \bar X_{t_{k}}}({\bf{a}})  \vert  \bar{X}_{t_{k-1}} =z \right),  
 \end{eqnarray*}
for every $k=m+1,\dots,n$  and setting 
\begin{equation} \label{def_kernel_Hnm}
 H_{m} f =  \mathds{E} \left[ f(\bar{X}_{t_m}) \right]
 \end{equation}
yields  for every $k=m+1,\dots,n$,
 \begin{eqnarray*}
(\pi_{k,m} f)(x) & =  &  \mathds{E} \Big(  \mathds{E} \Big( f(\bar{X}_{t_k}) \prod_{i=m}^{k-1}   G_{\Delta_k\sigma_k^2}^{\bar X_{t_{k}}, \bar X_{t_{k+1}}}({\bf{a}})  \big  \vert (\bar{X}_{t_{\ell}})_{\ell=m,\dots,k-1} \Big) \big \vert \bar{X}_{t_m} =x   \Big)  \\
&  = &  \mathds{E} \Big(  \mathds{E} \Big( f(\bar{X}_{t_k})    G_{\Delta_{k-1}\sigma_{k-1}^2}^{\bar X_{t_{k-1}}, \bar X_{t_{k}}}({\bf{a}})  \vert (\bar{X}_{t_{\ell}})_{\ell=m,\dots,k-1} \Big)  \prod_{\ell=m}^{k-2}  G_{\Delta_{\ell} \sigma_{\ell}^2}^{\bar X_{t_{\ell}}, \bar X_{t_{\ell+1}}}({\bf{a}}) \vert \bar{X}_{t_m}   =x \Big)  \\
&  = &  \mathds{E} \Big(  \mathds{E} \Big( f(\bar{X}_{t_k})     G_{\Delta_{k-1}\sigma_{k-1}^2}^{\bar X_{t_{k-1}}, \bar X_{t_{k}}}({\bf{a}}) \vert \bar{X}_{t_{k-1}} \Big)  \prod_{\ell =m}^{k-2}  G_{\Delta_{\ell} \sigma_{\ell}^2}^{\bar X_{t_{\ell}}, \bar X_{t_{\ell+1}}}({\bf{a}})  \vert \bar{X}_{t_m} =x  \Big)\\
& = & (\pi_{k-1,m} H_{k} f)(x).
\end{eqnarray*}
Consequently, if one has access to the recursive  quantizations $(\widehat X_{t_k})_{k=m,\dots,n}$ and the transition probabilities $\{ \hat p_{k}^{ij}, k=m+1,\dots,n \}$  of the process $(\bar X_{t_k})_{k=m,\dots,n}$,   the quantity  $F(t_m,t_n,x)$ will be estimated by 
\begin{equation}  \label{EqEstimationF(t_m,t_n,)}
\widehat F(t_m,t_n, x) = \sum_{j=1}^{N_n} \widehat{\pi} _{n,m} \,\delta_{\{x_m^j =x\}},
\end{equation}
where the $\widehat{\pi}_{n,m}$'s are defined from the following recursive formula
\begin{equation}
\left \{
\begin{array}{ll}
\widehat{\pi}_{m,m} = \widehat H_{m} \\
\widehat{\pi}_{k,m} = \widehat{\pi}_{k-1,m} \widehat H_{k}:=\Big[   \sum_{i=1}^{N_{k-1}} \widehat H_{k}^{i,j} \widehat{\pi}_{k-1,m}  \Big]_{j=1,\dots,N_k}, k=m+1,\dots,n 
\end{array}
\right.
\end{equation}
 with 
$$ \widehat{H}_{k}^{ij} = G_{\Delta_{k-1}\sigma_{k-1}^2}^{x^i_{k-1}, x^j_k }({\bf{a}})  \,  \hat{p}_k^{ij}\, \delta_{x_k^j}, \quad i=1,\dots,N_{k-1}; j=1,\dots,N_k.$$

\subsection{Approximation of $\Pi_{y,m} \bar F(t_m,t_n,\cdot)$ by recursive quantization}
Combining  equations (\ref{EqPiF(t_n,t_m,)}) and (\ref{EqEstimationF(t_m,t_n,)}),   the conditional survival probability $\Pi_{y,m} F(t_m,t_n,\cdot)$ will be estimated (for a fixed  trajectory $(y_0,\dots,y_m)$ of the observation process $(Y_{t_0},\dots,Y_{t_m})$) by
\begin{equation}  \label{EqFiniteApproxForm}
  \widehat{\Pi}_{y,m} \widehat F(t_m,t_n,\cdot)  = \sum_{i=1}^{N_m}  \sum_{j=1}^{N_n}   \widehat{\Pi}^i_{y,m}  \widehat{\pi} _{n,m}\, \delta_{x_m^j} .  
\end{equation} 
\begin{rem}\label{rem:AbbCal}
In Section \ref{sec:numerics}, the formula \eqref{EqFiniteApproxForm} will be compared to the one  of interest in  \cite{ProSag14}: $\mathbb P \big(\tau_{\bar X} > t \vert {\cal F}_s^{\bar  Y} \big)$, which reads (following the previous notations) 
\begin{equation}\label{eq:condidionalsurvY}
   \mathbb P \big(\tau_{\bar X} > t \vert\, (\bar Y_{t_0}, \ldots, \bar Y_{t_m}) =y \big) = \frac{\mathbb E \big[\bar{F}(t_m,t_n,\bar X_{t_m})  L_{ y,{\bf{a}}}^m \big] }{ \mathbb E[   L_y^m]}.  
\end{equation}
The conditional probability has been approximated in \cite{ProSag14} via an hybrid Monte Carlo - optimal quantization method. It may be approximated following the procedure we propose using only optimal quantization method as 

\begin{equation}\label{eq:FinapproxY}
    \widehat{\varpi}_{y,m} \widehat F(t_m,t_n,\cdot)  = \sum_{i=1}^{N_m}  \sum_{j=1}^{N_n}   \widehat{\varpi}^i_{y,m}  \widehat{\pi} _{n,m}\, \delta_{x_m^j} 
\end{equation}
where the $\widehat{\varpi}^i_{y,m}$'s are obtained from  \eqref{eq:hatsmallpi} by replacing the function $g^{\bf{a}}_k$ by $g_k$ of equation \eqref{EqDefFunctg_k}.
\end{rem}

Let us say now how to quantize the signal process  $X$ from the recursive quantization method.

\subsection{The recursive quantization method}

Recall first that for  a given $\mathds R^d$-valued random vector $X$ defined on  $(\Omega,\mathcal{F},\mathbb{P})$ with distribution $\mathbb{P}_X$,  the $L^r(\mathbb P_{X})$-optimal quantization  problem of size $N$ for  $X$ (or for  the distribution $\mathbb P_X$)  aims to approximate  $X$  by a Borel function   of  $X$  taking  at most  $N$ values. If  $X \in L^r(\mathbb{P})$ and defining ${\Vert X \Vert}_r := {\left(\mathbb E \vert X \vert ^r \right)}^{1/r}$  where $ \vert \cdot \vert $ denotes an arbitrary  norm on $\mathbb{R}^d$,   this turns out to  solve the following optimization problem (see e.g. \cite{GraLus}):
\begin{equation}
 e_{N,r}(X)  = \inf{\{ \Vert X - \widehat{X}^{\Gamma} \Vert_r, \Gamma \subset \mathbb{R}^d, \textrm{card}(\Gamma) \leq N \}} =  \inf_{ \substack{\Gamma  \subset \mathbb{R}^d \\ \textrm{card}(\Gamma) \leq N}} \left(\int_{\mathbb{R}^d} d(x,\Gamma)^r d\mathbb P_X(x) \right)^{1/r}\label{er.quant}
 \end{equation}
where $\widehat{X}^{\Gamma}$,  the   quantization of $X$  on  the subset   $\Gamma = \{x_1,\ldots,x_N  \}\subset \mathbb{R}^d$ (called a codebook, an $N$-quantizer or a grid) is defined by   
$$ \widehat{X}^{\Gamma} = {\rm Proj}_{\Gamma}(X) := \sum_{i=1}^N  x_i  \mathds{1}_{\{X \in C_i(\Gamma)\}}$$ 
 and where      $(C_i(\Gamma))_{ i=1,\ldots,N}$   is a Borel partition (Voronoi partition) of  $\mathbb R^d$  satisfying for every $i \in \{1,\ldots,N\}$,  $$ C_i(\Gamma) \subset \{ x \in \mathbb{R}^d : \vert x-x_i \vert = \min_{j=1,\ldots,N}\vert x-x_j \vert \}.$$

 Keep in mind that for every  $N \geq 1$, the infimum in $(\ref{er.quant})$ is reached at    one  grid at least. Any $N$-quantizer realizing this infimum  is called an  $L^r$-optimal $N$-quantizer. Moreover, if  $\textrm{ card(supp}(\mathbb P_X)) \geq N$ then  the optimal $N$-quantizer is of size   $N$ (see \cite{GraLus} or \cite{Pag}).  
 On the other hand,    the   quantization error, $e_{N,r}(X)$, decreases to zero  at  an  $N^{-1/d}$-rate  as the grid size $N$ goes to infinity. This convergence rate (known as Zador   Theorem) has been investigated in~\cite{BucWis} and~\cite{Zad}  for absolutely continuous probability measures under the  quadratic norm on $\mathbb{R}^d$. A detailed study of the convergence rate under an arbitrary norm on  $\mathbb{R}^d$  and for both   absolutely continuous  and singular measures  may be found in \cite{GraLus}.

 The recursive  quantization of the Euler scheme of an $\mathbb R^d$-valued  diffusion process  has been introduced in \cite{PagSagMQ}.  The method allows to speak of fast online quantization and consists on  a sequence of   quantizations $(\widehat X_{t_k}^{\Gamma_k}  )_{k=0,\ldots,n}$ of the Euler scheme   $(\bar X_{t_k}  )_{k=0,\ldots,n}$  defined recursively as 
 \begin{equation}
  \widetilde X_0  =\bar X_0\;,\quad
\widehat X_{t_k}^{\Gamma_k}  = {\rm Proj}_{\Gamma_k}(\widetilde X_{t_k})    \quad \textrm{and}\quad   \widetilde X_{t_{k+1}} = {\cal E}_k(\widehat X_{t_k}^{\Gamma_k} ,Z_{k+1}) ,\; k=0,\ldots,n-1, \label{EqGenTildeX1}
  \end{equation}
where $(Z_k)_{k=1,\ldots, n}$ is an  i.i.d. sequence of  ${\cal N}(0;I_q)$-distributed random vectors, independent from  $\bar X_0$ and 
\[
   {\cal E}_{k}(y,z)  = y + \Delta b(t_{k},y) + \sqrt{\Delta} \sigma(t_{k},y) z, \quad y \in \mathbb R^d, \ z \in \mathbb R^q, \ k=0, \ldots, n-1.
   \]
  The sequence of quantizers satisfies for every $k \in \{0, \ldots,n\}$,
\[
 \Gamma_k \!\in  \argmin \{ \widetilde D_{k}(\Gamma),\  \Gamma \subset \mathbb R^d, \ {\rm card}(\Gamma) \leq N_{k} \},
 \]
where for every grid $\Gamma \subset \mathbb R^d$, $\widetilde D_{k+1} (\Gamma)  :=  \mathbb E \big[   {\rm dist}(\widetilde X_{t_{k+1}}, \Gamma)^2 \big]$.

 This  recursive quantization method  raises some   problems among which the computation of the  quadratic error bound  $\Vert \bar X_{t_k}  - \widehat X_{t_k}^{\Gamma_k}  \Vert_{_2} : = \big(\mathbb E\vert  \bar X_{t_k} - \widehat X_{t_k}^{\Gamma_k}  \vert_{_2}\big)^{1/2}$, for every $k=0, \ldots,n$.  It has been shown in \cite{PagSagMQ, PagSagMQJump} that for any sequences of (quadratic) optimal quantizers  $\Gamma_k$ for $\widetilde X_{t_k}^{\Gamma_k} $, for  every $k=0,\dots,n-1$, the quantization error  $\Vert \bar X_{t_{k}} -  \widehat X_{t_{k}}^{\Gamma_k}  \Vert_{_2}$ is bounded by the cumulative quantization errors  $\Vert \widetilde X_{t_{i}} -  \widehat X_{t_{i}}^{\Gamma_i}  \Vert_{_2}$, for $i=0, \ldots,k$. This result is obtained under the following assumptions and is stated in Proposition \ref{PropGlobBound} below:
 
\begin{enumerate}
\item {\em $L^2$-Lipschitz assumption}. The mappings $x\mapsto {\cal E}_k(x,Z_{k+1})$  from $\mathbb R^d$ to $L^2(\Omega,{\cal A}, \mathbb P)$, $k=1:n$ are Lipschitz continuous~i.e.  
\[
({\rm Lip})\quad \equiv \quad\forall\,  x,\, x' \!\in \mathbb R^d,\quad  \big\|{\cal E}_k(x,Z_{k+1})-{\cal E}_k(x',Z_{k+1})\big\|_{_2}\le [{\cal E}_k]_{\rm Lip} |x-x'|,\; k=1:n.
\]  

\item {\em $L^p$-linear growth assumption}.  Let    $\, p\!\in (2,3]$.
\[
({\rm SL})_p\; \equiv  \; \forall\, k\!\in \{1,\ldots, n\},\; \forall\,x \!\in \mathbb R^d, \quad  \mathbb E\vert {\cal E}_k(x,Z_{k+1}) \vert^p\le \alpha_{p,k}+\beta_{p,k} |x|^p.
\]
\end{enumerate}

 \begin{prop} \label{PropGlobBound} Let $\widehat X= (\widehat X_{t_k})_{k=0:n}$ be defined by~\eqref{EqGenTildeX1} and suppose that all the grids $\Gamma_k$  are quadratic optimal. Assume  that both assumptions $({\rm Lip})$ and  $({\rm SL})_p$ (for some $p \in (2,3]$)  hold and that $X_0\!\in L^p(\mathbb{P})$.  Then, 
\begin{equation}  \label{EqRecQuantBound}
\big\| \bar X_{t_k}-\widehat X_{t_k}\big\|_{_2} \le C_{d, p} \sum_{i=0}^k   [{\cal E}_{i+1:k}]_{\rm Lip}  \left[ \sum_{\ell=0}^i  \alpha_{p,\ell}\beta_{p,\ell+1:i}\right] ^{\frac 1p} \, N_i^{-\frac 1d}
\end{equation}
where $C_{d, p}>0$ and     $\alpha_{p,0}=\mathbb E\, |X_0|^p = \Vert  X_0 \Vert^{^p}_{_p}$, \ $\beta_{p,\ell:i}=\prod_{m=\ell}^{i}\beta_{p,m}$ (with $\prod_{\emptyset}= 1$) and  
\[
  [{\cal E}_{i:k}]_{\rm Lip}:=  \prod_{\ell =i}^k [{\cal E}_{\ell}]_{\rm Lip},\; 1 \le \ell \le  k\le n \qquad \mbox{ and } \quad  [{\cal E}_{k+1:k}]_{\rm Lip}=1 .
\]
\end{prop}

The associated probability weights and transition probabilities are computed from explicit formulas we recall in the following result. 

\begin{prop}
Let $\Gamma_{k+1}$ be a quadratic optimal quantizer for the marginal random variable $\widetilde X_{t_{k+1}}$. Suppose that the quadratic optimal quantizer $\Gamma_k$ for $\widetilde X_{t_k}$ is already computed and that we have access to  its associated  weights $\mathbb P(\widetilde X_{t_k} \in C_i(\Gamma_k))$, $i=1,\ldots,N_k$.  The transition probability $\hat p_k^{ij} = \mathbb P(\widetilde X_{t_{k+1}} \in C_j(\Gamma_{k+1})\vert \widetilde X_{t_k}  \in C_i(\Gamma_{k}))$ $= \mathbb P(\widehat X_{t_{k+1}} =x_{k+1}^j\vert \widehat X_{t_k}  = x_k^i)$ is given by
\begin{eqnarray}\label{eq:transiprob}
\hat p_k^{ij}   & =&  \Phi\big(x_{k+1,j+}(x^{k}_i)\big)  -     \Phi\big(x_{k+1,j-}(x^{k}_i)\big)   \label{EqEstProba},
\end{eqnarray}
 where $\Phi(\cdot)$ is the cumulative distribution function of the standard Gaussian distribution, 
 $$ x_{k+1,j-}(x): = \frac{x_{k+1}^{j-1/2} - m_{k}(x) }{v_k(x)} \quad \textrm{ and } \quad x_{k+1,j+}(x): = \frac{x_{k+1}^{j+1/2} - m_{k}(x) }{v_k(x)},$$
 with   $m_k(x) = x +\Delta b(t_k,x)$, $v_k(x) = \sqrt{\Delta} \sigma(t_k,x)$ and,  for $k=0,\ldots,n-1$ and for    $j=1,\ldots,N_{k+1}$,
$$ x_{k+1}^{j-1/2} = \frac{ x_{k+1}^{j} + x_{k+1}^{j-1}  }{2}, \  x_{k+1}^{j+1/2} = \frac{ x_{k+1}^{j} + x_{k+1}^{j+1}  }{2}, \ \textrm{ with }  x_{k+1}^{1/2} = -\infty,  x_{k+1}^{N_{k+1}+1/2} =+\infty.$$ 
 \end{prop}
Once we have access to the marginal quantizations and to its associated transition probabilities, the right hand side of \eqref{EqFiniteApproxForm} can be computed explicitly. 


\subsection{The error analysis}
Our aim in this section is to investigate the  error   resulting from the approximation of $\Pi_{y,m} \bar F(t_m, t_n,\cdot) = \big[\Pi_{y,m} \big(\pi_{n,m} \mbox{\bf 1}\big)\big](\cdot)$ by   $\big[ \widehat{\Pi}_{y,m}  \big(\hat{\pi}_{n,m} \mbox{\bf 1}\big)\big](\cdot)$. This error is an aggregation of three terms (see the proof of Theorem \ref{ThmErrorBounds}) involving the approximation errors $\vert \Pi_{y,m} \bar F(t_m, t_n,\cdot) - \widehat{\Pi}_{y,m} \bar F(t_m, t_n,\cdot) \vert$ and $\vert \big(\pi_{n,m} \mbox{\bf 1}\big)(x) -  \big(\hat{\pi}_{n,m} \mbox{\bf 1}\big)(x)\vert$, $x \in \mathbb R$. The two following results give the  bounds associated two the former approximation errors. Both are  (carefully) adjustments of Theorem 4.1. and Lemma 4.1. in \cite{PagPha} to our context so that we refer to the former paper for their detailed proofs.

    

\begin{thm} \label{ThmConvergence}    Suppose that Assumption ({\rm Lip}) holds true. Then, for any bounded  Lipschitz continuous function $f$ on $\mathds R^d$ we  have,
\begin{eqnarray*}
\vert   \Pi_{y,m}  f - \widehat{\Pi}_{y,m} f   \vert  &  \leq & \frac{ K^m_g }{\phi_m \vee  \hat{\phi}_m}  \sum_{k=0}^m  A_k^m(f,y) \ \Vert \bar{X}_{t_k} - \widehat{X}^{\Gamma_k}_{t_k}  \Vert_{_2},  
\end{eqnarray*}
where 
$$ \phi_{m} := \pi_{y,m} \mbox{\bf 1},  \quad \widehat{\phi}_{m} := \widehat{\pi}_{y,m}  \mbox{\bf 1}, $$  
\begin{eqnarray*}
A_k^m(f,y) & := &   2  \frac{\Vert f\Vert_{\infty}}{K_g^m} [g_{k}^2]_{\rm Lip}(y_{k-1},y_k)  +   2  \frac{\Vert f\Vert_{\infty}}{K_g^m}  \sum_{j=k+1}^m  [{\cal E}]_{Lip}^{j-k-1} \Big([g_{j}^1]_{\rm Lip}(y_{j-1},y_j) \\
&& \hspace{6.5cm} + \, [{\cal E}]_{\rm Lip} [g_{j}^2]_{Lip}(y_{j-1},y_j) \Big),
\end{eqnarray*}
and, for every $k \in \{1, \ldots,m\}$, $[g^1_k]_{\rm Lip}(y,y')$ and $[g^2_k]_{\rm Lip}(y,y')$ are such that  for every $x,x', \hat{x}, \hat{x}' \in \mathbb R^d$,
$$  \vert g_{k}^a(x,y;x',y') - g_{k}^a(\widehat{x},y;\widehat{x}',y') \vert  \leq [g^1_{k}]_{\rm Lip}(y,y')\ \vert x-\widehat{x} \vert + [g^2_{k}]_{\rm Lip}(y,y') \ \vert x'-\widehat{x}' \vert.$$
The quantities $K_g$ and $[{\cal E}]_{\rm Lip}$ are defined as
\[
K_g = \max_{k=1,\cdots,m} \Vert  g_{k}^a  \Vert_{\infty} \qquad \mbox{ and } \qquad [{\cal E}]_{\rm Lip} = \max_{k=1, \ldots, m} [{\cal E}_k]_{\rm Lip}.
\]
\end{thm}

Remark that the existence of $[g^1_k]_{\rm Lip}(y_{k-1},y_k)$ and $[g^2_k]_{\rm Lip}(y_{k-1},y_k)$  is guaranteed by the fact that the function $g_{k}^a(x,y;x',y')$ is Lipschitz with respect to $(x,x')$.
\medskip

Let us give now the error bound associated to the approximation of  $\pi_{n,m} \mbox{\bf 1}$.

\begin{prop}  \label{ProTh;conver} Let  $y = (y_0, \ldots, y_m) \in (\mathbb R^q)^{m+1}$. Then,  we have for any $x \in \mathbb R$
\begin{equation}. \label{EqQuantErrorPi}
    \big\vert \big(\pi_{n,m} \mbox{\bf 1} \big)(x) -  \big(\hat{\pi}_{n,m} \mbox{\bf 1} \big)(x) \big\vert \le \sum_{k=m+1}^n B_k (G)\ \Vert \bar{X}_{t_k} - \widehat{X}^{\Gamma_k}_{t_k}  \Vert_{_2}
\end{equation}
where 
\[
B_k(G) = \Lambda^{n-m-1} \Big([G_m^1]_{\rm Lip}\, \delta_{_{\{k=m \}}} + \big([G_k^1]_{\rm Lip} \vee [G_{k-1}^2]_{\rm Lip}\big) \,\delta_{_{\{k \in \{ m+1, \ldots, n-1\}\}}} + [G_n^2]_{\rm Lip}\, \delta_{_{\{k=n}\}} \Big)
\]
with 
\[
\Lambda = \max_{k=m+1, \ldots, n}\big \Vert G_{\Delta_k \sigma_k^2}^{({\tiny \bullet},{\tiny \bullet})}({\bf a}) \big \Vert_{\infty}
\]
\end{prop}

\begin{proof}
Recall that 
$$\big(\pi_{n,m} \mbox{\bf 1}\big)(x)=  \mathbb E \Big( \Lambda_m (\bar X_{t_{m:n}})  \big\vert\, \bar{X}_{t_m} =x \Big) $$
where for every $k \ge   m$,
\[
\Lambda_m (\bar X_{t_{m:k}}) := \prod_{\ell=m}^{k-1} G_{\Delta_k\sigma_k^2}^{\bar X_{t_{k}}, \bar X_{t_{k+1}}}({\bf{a}})
\]
with the convention that $\Lambda_m (\bar X_{t_{m:m}}) =1 $. 
Now, we have for any $k \ge m
$,
\begin{eqnarray*}
\Lambda_m (\bar X_{t_{m:k}}) - \Lambda_m (\hat X_{t_{m:k}}) & = & \Big( G_{\Delta_{k-1}\sigma_{k-1}^2}^{\bar X_{t_{k-1}}, \bar X_{t_{k}}}({\bf{a}}) -  G_{\Delta_{k-1}\sigma_{k-1}^2}^{\hat X_{t_{k-1}}, \hat X_{t_{k}}}({\bf{a}})\Big) \Lambda_m (\bar X_{t_{m:k-1}}) \\
  &&  \ + \ G_{\Delta_{k-1}\sigma_{k-1}^2}^{\hat X_{t_{k-1}}, \hat X_{t_{k}}}({\bf{a}}) \big(\Lambda_m (\bar X_{t_{m:k-1}})  - \Lambda_m (\hat X_{t_{m:k-1}})\big).
\end{eqnarray*}
Since the function $G_{\Delta_k \sigma_k^2}^{({\tiny \bullet},{\tiny \bullet})}({\bf a})$ is Lipschitz and bounded and that for any  $k \ge m+1$, $\Lambda_m (\bar X_{t_{m:k-1}}) \le \Lambda^{k-m-1}$, we have 
\begin{eqnarray*}
\vert \Lambda_m (\bar X_{t_{m:k}}) - \Lambda_m (\hat X_{t_{m:k}}) \vert &\le& \Big([G_k^1]_{\rm Lip} \vert \bar X_{t_{k-1}} - \hat X_{t_{k-1}}  \vert  + [G_k^2]_{\rm Lip} \vert \bar X_{t_{k}} - \hat X_{t_{k}}  \vert\Big) \Lambda^{k-m-1} \\
 && \ + \ \Lambda \vert \Lambda_m (\bar X_{t_{m:k-1}}) - \Lambda_m (\hat X_{t_{m:k-1}})  \vert. 
\end{eqnarray*}
Keeping in mind that $\Lambda_m (\bar X_{t_{m:m}}) = \Lambda_m (\hat X_{t_{m:m}}) = 1$, we deduce from an induction on $k$ that 
\[
\vert \Lambda_m (\bar X_{t_{m:n}}) - \Lambda_m (\hat X_{t_{m:n}}) \vert \le \Lambda^{n-m-1} \sum_{k=m+1}^n [G_k^1]_{\rm Lip} \vert \bar X_{t_{k-1}} - \hat X_{t_{k-1}}  \vert + [G_k^2]_{\rm Lip} \vert \bar X_{t_{k}} - \hat X_{t_{k}}  \vert. 
\]
The result follows by noting that  
 \begin{eqnarray*}
 \big\vert \big(\pi_{n,m} \mbox{\bf 1} \big)(x) -  \big(\hat{\pi}_{n,m} \mbox{\bf 1} \big)(x) \big\vert & \le & \mathbb E \big( \vert \Lambda_m (\bar X_{t_{m:n}}) - \Lambda_m (\hat X_{t_{m:n}})  \vert\,  \big\vert \bar X_{t_m} = x\big) \\
 & \le & \Lambda^{n-m-1} \sum_{k=m+1}^n [G_k^1]_{\rm Lip} \vert \bar X_{t_{k-1}} - \hat X_{t_{k-1}}  \vert + [G_k^2]_{\rm Lip} \vert \bar X_{t_{k}} - \hat X_{t_{k}}  \vert
 \end{eqnarray*}
 and by using the non-decreasing property of the $L^p$-norm.
\end{proof}

We may deduce now the global error induced by our procedure, means,   the error deriving from the estimation of  $$ \Pi_{y,m} F(t_m,t_n,\cdot) = \mathbb P\big(\tau_X>t_n  \vert\,  (Y_{t_0},\dots, Y_{t_m})  = y \big) $$ by Equation  \eqref{EqFiniteApproxForm}.  To this end, we need the following additional assumptions which will be used to compute (see \cite{GobThese})  the convergence rate  of the quantity  $\mathbb E \big\vert  \mathds{1}_{\{ \tau_{\bar X}>t \}} - \mathds {1}_{\{ \tau_X>t \}} \big\vert$ towards   $0$. We suppose that the diffusion is homogeneous and   
\begin{itemize}
\item [(\textbf {H1})] $b$ is a $\mathcal C_b^{\infty}(\mathbb R)$ function and $\sigma$ is in $\mathcal C_b^{\infty}(\mathbb R)$.
\item [(\textbf {H2})] There exists $\sigma_0>0$ such that $\forall x \in \mathbb R, \sigma(x)^2 \ge \sigma_0^2  $ (\emph {uniform ellipticity}).

\end{itemize}

We have the following result.

\begin{prop}[See \cite{GobThese}]  \label{PropGob}
Let   $t>0$.  Suppose that Assumptions  (\textbf{H1}) and (\textbf{H2}) are fulfilled. Then,   for every $\eta \in (0,\frac{1}{2}[$   there exists an increasing function $K(T)$ such that for every $t \in [0,T]$ and for every $x \in \mathbb R$,
$$
\mathbb E_x \left[\big\vert  \mathds{1}_{\{\tau_X > t \}}  - \mathds{1}_{\{\tau_{\bar X}> t\}}  \big\vert \right] \le \frac{1}{ n^{\frac{1}{2}-\eta}} \frac{K(T)}{\sqrt{t}},
$$
 where  $n$ is the number of discretization time steps over $[0,t]$. 
\end{prop}

\begin{thm} \label{ThmErrorBounds} 
Suppose   that the coefficients $b$ and $\sigma$ of the continuous signal process $X$ are such that Assumptions    \textbf{(H1)} and  \textbf{(H2)} are satisfied and let $\eta \in (0,\frac{1}{2}]$. We also suppose that Assumption ({\rm Lip}) holds. Then,  for any bounded Lipschitz continuous function $f:\mathbb R^d \mapsto \mathbb R$ and for any fixed observation $y  = (y_0, \ldots,y_m)$ we have 
\begin{eqnarray*}
\Big \vert \mathbb P\big(\tau_X>t_n \vert  (\bar Y_{t_0},\dots,\bar Y_{t_m})=y\big) &  - &   \sum_{i=1}^{N_n}  \sum_{j=1}^{N_m}   \widehat{\Pi}^i_{y,m}  \widehat{\pi} _{n,m}\, \delta_{x_m^j}      \Big \vert     \leq   \mathcal O \Big(n^{-\frac{1}{2}+\eta} \Big)    \\ 
& + &   \sum_{k=0}^n C_k^n(\bar F(s,t,\cdot),y) \ \Vert \bar{X}_{t_k} - \widehat{X}^{\Gamma_k}_{t_k}  \Vert_{_2},
\end{eqnarray*}
where 
$$C_k^n(\bar F(s,t,\cdot),y) = \frac{ K^m_g }{\phi_m \vee  \hat{\phi}_m}   A_k^m(f,y) \, \delta_{\{k \le m \}} + B_k(G) \, \delta_{\{k \ge m+1 \}}$$
and where $K^m_g$, $\phi_m$,  $\hat{\phi}_m$,    $A_k^m(f,y)$ and  $B_k(G)$ are defined in  Theorem \ref{ThmConvergence} and in Proposition \ref{ProTh;conver}.
\end{thm}

\begin{proof}[$\textbf{Proof}$] We have
\begin{eqnarray} \label{EqDecomposError}
 \Big \vert \mathbb P(\tau_X>t_n \vert  (\bar Y_{t_0},\dots,\bar Y_{t_m})=y) &  - &       \sum_{i=1}^{N_n}  \sum_{j=1}^{N_m}   \widehat{\Pi}^i_{y,m}  \widehat{\pi} _{n,m} \, \delta_{x_m^j}    \Big \vert   \nonumber  \\ 
&  \leq & \big \vert \mathbb P(\tau_X>t_n \vert  (\bar Y_{t_0},\dots,\bar Y_{t_m})=y) -  \mathbb P(\tau_{\bar X}>t_n \vert ( \bar Y_{t_0},\dots, \bar Y_{t_m})=y) \big \vert  \nonumber  \label{EqDiscretError}  \\
& &  + \ \big \vert   \Pi_{y,m} \bar F(t_m,t_n,\cdot) - \widehat{\Pi}_{y,m} \bar F (t_m,t_n,\cdot) \big \vert   \nonumber \\
&  & +\  \big\vert  \widehat{\Pi}_{y,m} \bar F (t_m,t_n,\cdot)  -   \widehat{\Pi}_{y,m}  \hat  F(t_m,t_n,\cdot)  \big \vert. 
\end{eqnarray}
Now, we have 
\begin{eqnarray*}
 \big\vert  \widehat{\Pi}_{y,m} \bar F (t_m,t_n,\cdot)  -   \widehat{\Pi}_{y,m}  \hat  F(t_m,t_n,\cdot)  \big \vert & =&  \Big\vert \sum_{i=1}^{N_m} \Big( \widehat{\Pi}_{y,m}^i \bar F (t_m,t_n,x_m^i)  -   \widehat{\Pi}_{y,m}^i  \hat  F(t_m,t_n,x_m^i) \Big)  \Big \vert \\
 & \le & \sup_{x \in \mathbb R}  \big\vert \bar F (t_m,t_n,x)  -   \hat  F(t_m,t_n,x)  \big\vert    \sum_{i=1}^{N_m}  \widehat{\Pi}_{y,m}^i \\
 &  = &  \sup_{x \in \mathbb R}  \big\vert \bar F (t_m,t_n,x)  -   \hat  F(t_m,t_n,x)  \big\vert   \\
& = &  \sup_{x \in \mathbb R}  \big\vert  (\pi_{n,m} \mbox{\bf 1})(x)  -   (\hat   \pi_{n,m} \mbox{\bf 1})(x)  \big\vert.  
\end{eqnarray*}
The error bound for  $\vert  (\pi_{n,m} \mbox{\bf 1})(x)  -   (\hat   \pi_{n,m} \mbox{\bf 1})(x)  \vert$ is independent from $x$ and is given by \eqref{EqQuantErrorPi}. On the other hand we have 
\begin{eqnarray*}
 \big \vert \mathbb P(\tau_X>t_n \vert  (\bar Y_{t_0},\dots,\bar Y_{t_m})=y) &-&  \mathbb P(\tau_{\bar X}>t_n \vert  (\bar Y_{t_0},\dots, \bar Y_{t_m})=y \big \vert  \\ 
&\le& \mathbb E \big( \vert \mathds{1}_{\{\tau_X>t_n\}} - \mathds{1}_{\{\tau_{\bar X}>t_n\}} \vert \big \vert  (\bar Y_{t_0},\dots, \bar Y_{t_m}=y)\big) \\
&\le& \frac{1}{\mathbb P ((\bar Y_{t_0},\dots, \bar Y_{t_m})=y)}\,  \mathbb E  \big \vert \mathds{1}_{\{\tau_X>t_n\}} - \mathds{1}_{\{\tau_{\bar X}>t_n\}}  \big \vert.  
\end{eqnarray*}
We conclude by Proposition \ref{PropGob}.
\end{proof} 

\section{Numerical results}\label{sec:numerics}
We illustrate the numerical part by considering a Black-Scholes model. More specifically, the dynamics of the signal process $X$ and observation process $Y$ are in the form

\begin{equation} \label{EqSignalStatePr1}
 \begin{cases}
 dX_t = X_t(\mu dt + \sigma dW_t),  &   X_0=x_0, \\
 dY_t = Y_t(\mu dt + \sigma dW_t +  \delta d\widetilde{W}_t), &  Y_0=y_0
 \end{cases} 
 \end{equation}
meaning that
\[
\frac{dY_t}{Y_t} = \frac{dX_t}{X_t} + \delta d\widetilde{W}_t
\]
which can be interpreted as the yields of the observation process $Y$ are noised yields of the signal process $X$ with magnitude $\delta$. Intuitively, in order to deal with CDS option implied volatilities below, we will play with the parameter $\delta$.

\subsection{Comparison of conditional survival probabilities}
Before tackling the CDS examples, we first test the performance of our method in two setups:
\begin{enumerate}
\item[-] By comparing   the function $F(s,t,x)=\mathbb{P}(\inf_{s\leq u\leq t}X_u>{\bf{a}}|X_s=x)$ and its quantized version $\widehat F(s,t,x)$ defined in \eqref{EqEstimationF(t_m,t_n,)}, keeping in mind that in the model \eqref{EqSignalStatePr1}, 
\begin{equation}\label{eq:F_exact}
    F(s,t,x)= \Phi(h_1(x,t-s)) - \left(\frac{{\bf{a}}}{x}\right)^{2\sigma^{-2}(\mu-\sigma^2/2)} \Phi(h_2(x,t-s)),
\end{equation}
where
\begin{eqnarray*}
h_1(x,u) &=& \frac{1}{\sigma\sqrt{u}}\left(\log\left(\frac{x}{{\bf{a}}}\right) + \left(\mu-\frac{1}{2}\sigma^2\right)u\right),\\
h_2(x,u) &=& \frac{1}{\sigma\sqrt{u}}\left(\log\left(\frac{{\bf{a}}}{x}\right) + \left(\mu-\frac{1}{2}\sigma^2\right)u\right)\;.
\end{eqnarray*}
This comparison allows us to test our method given benchmark values.  
 \item[-] By comparing the conditional default probabilities $\mathbb P(\tau_{\bar X} \le t \vert {\cal F}_s^{\bar Y} \vee {\cal F}_s^{\bar H})$ and  $\mathbb P(\tau_{\bar X}\le t \vert {\cal F}_s^{\bar Y})$ respectively estimated by \eqref{EqFiniteApproxForm} and \eqref{eq:FinapproxY}, where $F(s,t,\cdot)$ is computed using the exact formula \eqref{eq:F_exact}. Our aim here is to confirm the impact of the additional information ${\cal F}_s^{\bar H}$ on the conditional probability $\mathbb P(\tau_{\bar X}\le t \vert {\cal F}_s^{\bar Y})$.
\end{enumerate}
 
 Notice that we deal in this paper with a general framework where the signal and the observation processes have no closed formula. Even if in our model both processes have explicit solutions and we may deduce a similar formula to  \eqref{EqExplicitCompSurvivalProba} using these closed formulae (the only change will come from the function $g_k$ which involves the  conditional density of $(X_{t_{k+1}},Y_{t_{k+1}})$ given $(X_{t_{k}},Y_{t_{k}})$), we still consider their associated  Euler schemes processes $\bar X$ and $\bar Y$ in order to stay in the scope of the proposed numerical method.
 
 To compare the functions $F(s,t,\cdot)$  and $\widehat F(s,t,\cdot)$, we choose the following set of parameters (like those of \cite{ProSag14}):$\mu=0.03$, $\sigma=0.09$, $\delta=0.5$, $x_0=y_0=86.3$ and ${\bf a}=76$. Figure \ref{fig:F_bar} shows the convergence of the quantized function $\widehat{F}(t_m,t_n,x)$ toward the exact one $F(t_m,t_n,x)$  with $t_m=1$, $t_n\in[1.1,3]$ and where $x$ is one point, say $x^\star_m$, on the grid $\{x_m^i, i=1, \ldots, N_m\}$ (see equation \eqref{EqEstimationF(t_m,t_n,)}). Once we fix $t_m$, $F(t_m,t_n,x)$ depends on both $t_n$ and $x$. 
 Therefore, to show the convergence, we fix $x$ 
 and plot both $F(t_m,t_n,x)$ and $\widehat{F}(t_m,t_n,x)$ with respect to $t_n$. The number of discretization points $m$ is set to $50$ and the convergence is achieved by increasing the number of quantization points $N_n$. Since the fixed quantization point $x^\star_m$ can differ when moving $N_n$, the corresponding figures can take different shapes but, we have only to make sure that the convergence is achieved when increasing $N_n$.\\
\begin{figure}[H]
\centering
\subfigure[$N_n=50$]{\includegraphics[width=0.48\columnwidth]{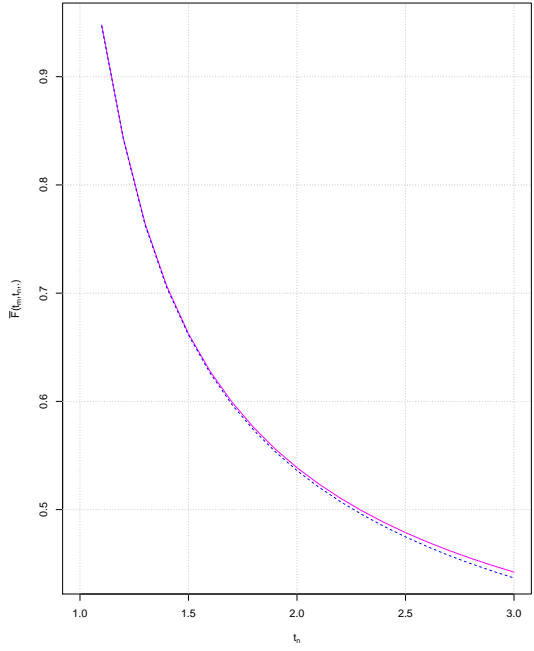}}\hspace{0.2cm}
\subfigure[$N_n=100$]{\includegraphics[width=0.48\columnwidth]{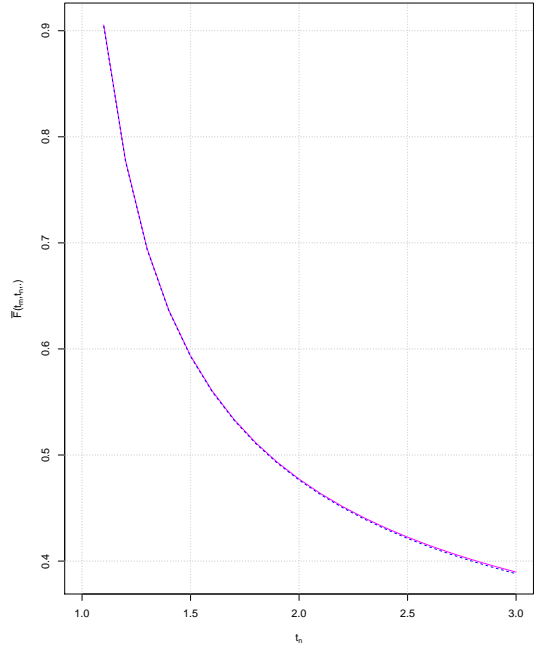}}\hspace{0.2cm}
\subfigure[$N_n=50$]{\includegraphics[width=0.48\columnwidth]{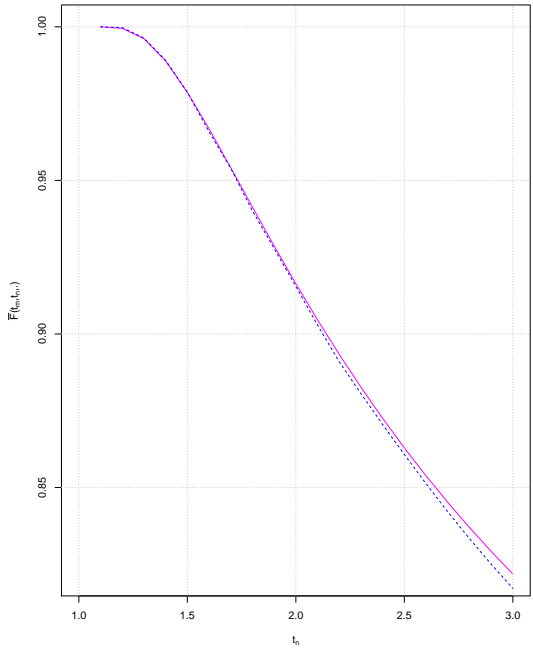}}\hspace{0.2cm}
\subfigure[$N_n=400$]{\includegraphics[width=0.48\columnwidth]{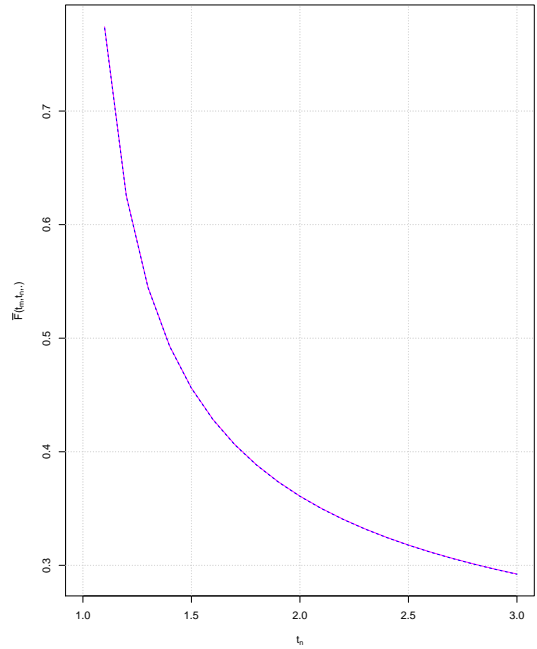}}\hspace{0.2cm}
\caption{Convergence of $\Bar{F}(t_m,t_n,\cdot)$ estimated by $\widehat{F}(t_m,t_n,\cdot)$ (dotted blue) toward $F(t_m,t_n,\cdot)$ (solid magenta).}\label{fig:F_bar} 
\end{figure}

 We now proceed to the numerical comparison between the conditional default probabilities $\mathbb P(\tau_{\bar X} \le t \vert {\cal F}_s^{\bar Y} \vee {\cal F}_s^{\bar H})$ and  $\mathbb P(\tau_{\bar X}\le t \vert {\cal F}_s^{\bar Y})$, respectively estimated by \eqref{EqFiniteApproxForm} and \eqref{eq:FinapproxY}, in order to check the statements of Remark \ref{rem:AbbCal}. 
 Setting $s=t_m=1$, $t=t_n$ and considering the same parameter set as in the previous figure but with $t_n\in[1.1,11]$, Figure \ref{fig:defaultProb} depicts the trajectories of the observation process $\bar Y$ from $0$ to $t_m$ and the associated conditional default probabilities as a function of $t_n$. First, we notice that equation \eqref{eq:defaultProbComp} is fulfilled as given a trajectory of the observation process $\bar Y$ represented in red, $\mathbb P(\tau_{\bar X}\le t \vert {\cal F}_s^{\bar Y})$ lined up in dots is always above $\mathbb P(\tau_{\bar X} \le t \vert {\cal F}_s^{\bar Y} \vee {\cal F}_s^{\bar H})$ in magenta. Second, the gab between the two quantities is larger for an downward movement of $Y$ compared to an upward movement for which the firm is less exposed to default. This can be understood by the fact that the more the firm is creditworthy, the less the default information is important and the less the default probability is. Then, the model preserves the memory of all the observed path of the process $Y$ when computing default probabilities. This path-dependent future of the default probabilities has already been shown in \cite{CocGemJea} and is known to be very important as it is implicit in reduced-form models for calibration purpose to historical data.

\begin{figure}[H]
\centering
\subfigure[$\bar Y$ Up]{\includegraphics[width=0.48\columnwidth]{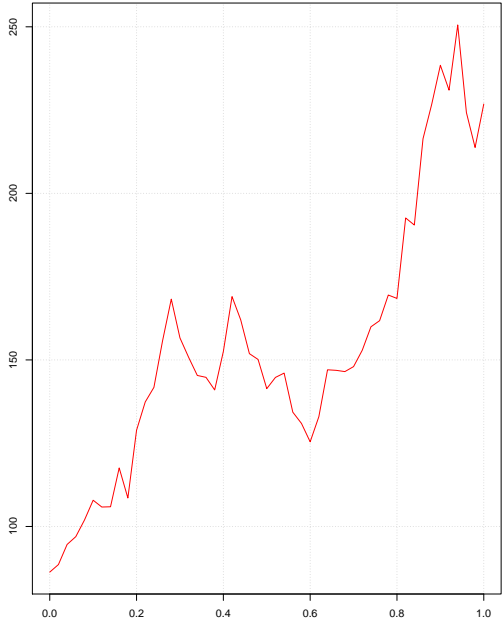}}\hspace{0.02cm}
\subfigure[Default probability]{\includegraphics[width=0.48\columnwidth]{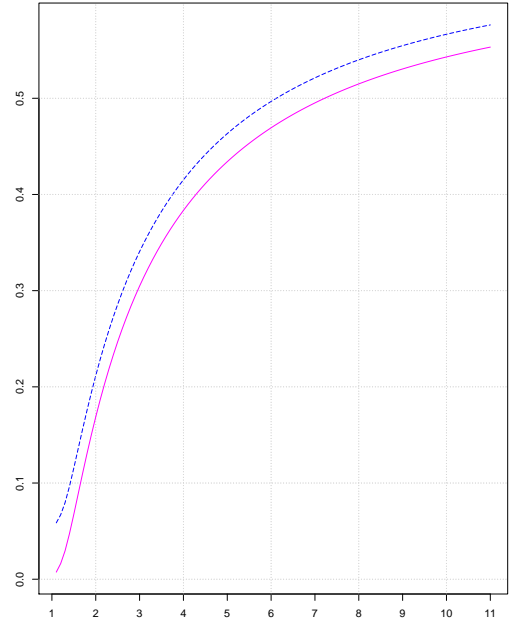}}\hspace{0.02cm}
\subfigure[$\bar Y$ Down]{\includegraphics[width=0.48\columnwidth]{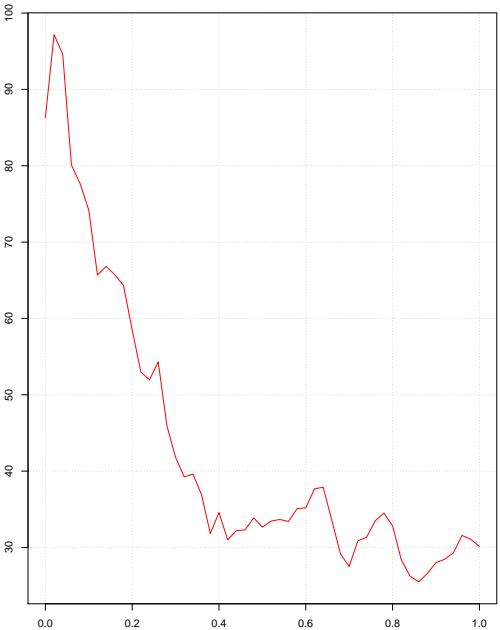}}\hspace{0.02cm}
\subfigure[Default probability]{\includegraphics[width=0.48\columnwidth]{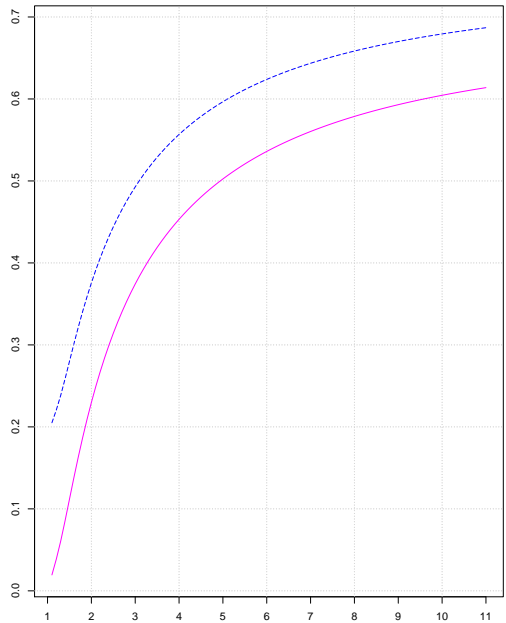}}\hspace{0.02cm}
\caption{Trajectories of the observation process $\bar Y$ (solid red), $Y$ up: panel (a), $\bar Y$ Down: panel (c) and the associated conditional default probabilities functions $\mathbb P(\tau_{\bar X} \le t \vert {\cal F}_s^{\bar Y} \vee {\cal F}_s^{\bar H})$ (solid magenta) and  $\mathbb P(\tau_{\bar X}\le t \vert {\cal F}_s^{\bar Y})$ (dotted blue): panels (b) and (d), with $t_m=1$ and $t_n\in[1.1,11]$ and number of quantization points $N_n=N_m=30$.}\label{fig:defaultProb} 
\end{figure}

\subsection{Application to CDS option pricing}
In this section, we briefly recall the concept and valuation  of credit default swaps and swaptions before analyzing the quantization procedure applied to such models. This will allow to give a full pricing formula of credit swaps derivatives in a firm value approach using partial information theory and optimal quantization. In addition, the fact that we add the default filtration in the model indicating whether default has already taken place or not is very important in this case as it is pointless to price a default swap post-default.\\
A credit default swap (CDS) is an agreement between two counterparties to buy or sell protection against the default risk of a third party called \emph{reference entity}. We set $\tau_X$ as the default time of the latter. In this case, if the contract is signed at time $s$, started at time $T_a$ with maturity $T_b$, the protection buyer pays a coupon (or spread) $k$ at payments dates $T_{a+1},\ldots,T_b$ as long as the reference entity does not default or until $\tau_X$. If the default occurs at time $\tau_X$ with $T_a<\tau_X\leq T_b$, the  protection seller will make a single payment $LGD$ (that we assume to be a known constant) to the protection buyer. A CDS option (CDSO) or default swaption is an option written on a default swap. From this perspective, it requires to recall the no-arbitrage pricing equation of a CDS. The time-$s$ price of a general buyer CDS $CDS_s(a,b,k)$ with unit notional starting at time $T_a$ with maturity $T_b$, $s\leq T_a<T_b$, a spread $k$ and loss given default $LGD$ is given by the difference of the conditional risk-neutral expectations of the protection and the premium discounted cashflows:
\begin{eqnarray}
CDS_s(a,b,k)&=&\mathbb{E}\left[LGD\mathds{1}_{\{T_a< \tau_X\leq T_b\}}D_s(\tau_X)|\mathcal{F}_s\right]\nonumber\\
&&-k\mathbb{E}\left[\left.\sum_{i=a+1}^b \left(\mathds{1}_{\{\tau_X\geq T_i\}}\alpha_i D_t(T_i)+ \mathds{1}_{\{T_{i-1}\leq \tau_X< T_i\}}\alpha_i\frac{\tau_X-T_{i-1}}{T_i-T_{i-1}} D_t(\tau_X)\right)\right|\mathcal{F}_s\right]\nonumber
\end{eqnarray}
with $\alpha_i$ the day count fraction between dates $T_{i-1}$ and $T_i$ which, in a standard CDS, is around $0.25$ (quarterly payment dates) and $D_s(t)=e^{-\int_s^tr(u)du}$ is a time-$s$ discount factor with maturity $t$ and deterministic interest rates $r$. In a reduced-form setup, where $\mathcal F_s := \mathcal F_s^Y \vee \mathcal F_s^H$, this expression can be developed explicitly thanks to the Key lemma:
\begin{equation}
CDS_s(a,b,k)=\mathds{1}_{\{\tau_X >s\}}\left(-LGD\int_{T_a}^{T_b}D_s(u)\partial_uP_s(u)du -k~C_s(a,b)\right)\;,\label{eq:CDS}
\end{equation}
where \begin{equation}
    P_s(t):=\frac{S_s(t)}{S_s(s)}
\end{equation}
and $S_s(t)=\mathbb P \left( \tau_X \ge t  \Big\vert \mathcal F_s^Y\right)$ is known as the \emph{Az\'ema supermartingale} and $C_s(a,b)$ is the risky duration, i.e. the time-$s$ value of the CDS premia paid during the life of the contract when the spread is 1:
\begin{equation}
    C_s(a,b):=\sum_{i=a+1}^b\alpha_iD_s(T_i)P_s(T_i)-\int_{T_{i-1}}^{T_i}\frac{u-T_{i-1}}{T_{i}-T_{i-1}}\alpha_iD_s(u)\partial_uP_s(u)du\;.\nonumber
\end{equation}

The spread which, at time $s$, sets the forward start CDS at 0, called \emph{par spread},  is given by: 
\begin{equation}
    \mathds{1}_{\{\tau_X >s\}}k^\star_s(a,b):=\mathds{1}_{\{\tau_X >s\}}\frac{-LGD\int_{T_a}^{T_b}D_s(u)\partial_uP_s(u)du}{C_s(a,b)}\;.\label{eq:ParCDS}
\end{equation}

The no-arbitrage price of a call option on such a contract at time $s=0$ becomes
\begin{eqnarray}
PSO(a,b,k)&=&\mathbb{E}\left[(CDS_{T_a}(a,b,k))^+D_0(T_a)\right]\nonumber\\&=& D_0(T_a)\mathbb{E}\left[S_{T_a}(T_a)\left( LGD-\sum_{i=a+1}^b\int_{T_{i-1}}^{T_i}g_i(u)D_{T_a}(u)P_{T_a}(u)du\right)^+ \right]\;\label{eq:CDSO}
\end{eqnarray}
where $g_i(u):= LGD(r(u)+\delta_{T_b}(u))+k\frac{\alpha_i }{T_i-T_{i-1}}(1-(u-T_{i-1})r(u))$, with $\delta_{s}(u)$ the Dirac delta function centered at $s$.\\
The random terms inside the expectation \eqref{eq:CDSO} mainly the survival processes $S_{T_a}(\cdot)$ and $P_{T_a}(\cdot)$ are ready to be computed using optimal quantization. To do so, using equations \eqref{EqFiniteApproxForm} and \eqref{eq:FinapproxY}, one only needs to set 
\begin{equation}
  P_{T_a}(u)=\widehat{\Pi}_{y,a} \widehat F(T_a,u,\cdot)\quad  \text{and} \quad S_{T_a}(T_a) =\widehat{\varpi}_{y,a}\mbox{\bf{1}}\;.
\end{equation}
Hence the randomness in the expectation \eqref{eq:CDSO} is only from the observation process $Y$ simulated from time $0$ to $T_a$. This means that we should not need a lot of paths when estimating the expectation \eqref{eq:CDSO} using Monte Carlo simulation after computing the above mentioned survival processes using optimal quantization. This in turn motivates to fully estimate \eqref{eq:CDSO} using a hybrid Monte Carlo-optimal quantization procedure.\medskip

A CDS option has little liquidity but, just like usual equity options, is quoted in term of its Black implied volatility $\bar{\sigma}$ which is based on the assumption that the credit spread follows a geometric Brownian motion.\footnote{Recall that this does not mean in any way that the market naively believes that credit spreads exhibit log-normal dynamics. Market participants simply rely on the Black-Scholes machinery to convert a price into a quantity that is more intuitive to traders, namely implied volatilities.}

The Black formula for payer swaptions at time 0 with maturity $T_a$ is
\begin{equation}
PSO^{Black}(a,b,k,\bar{\sigma}) = C_0(a,b)\left[k^\star_0(a,b)\Phi(d_1)-k\Phi(d_2)\right]\nonumber
\end{equation}
where
\begin{equation}
d_1=\frac{\ln\frac{k^\star_0(a,b)}{k}+\frac{1}{2}\bar{\sigma}^2T_a}{\bar{\sigma}\sqrt{T_a}},\qquad d_2=d_1-\bar{\sigma}\sqrt{T_a}\;.\nonumber
\end{equation}

Hence, the CDS option implied volatility $\bar{\sigma}$ can be found by solving the following equation
\begin{equation}
    PSO(a,b,k)=PSO^{Black}(a,b,k,\bar{\sigma})\;.\nonumber
\end{equation}

We now assess the numerical results based on the model's applications to the pricing of CDS option. The model's parameter set is the same as before except here we take $\sigma=5\%$ and $\delta$ is varying. Table \ref{tab:CDSO} shows the estimated values of a European payer CDS option and the corresponding Black's volatilities with different strikes and different values of $\delta$. First, we observe that both CDS option prices and the implied volatilities are increasing with the noise volatility, $\delta$. This can be explained by the fact that, the higher $\delta$, the noisier the observations are and the higher the default probability. Since $\delta$ measures the degree of transparency of the firm, this will have a positive impact on the prices of the CDS option, hence on the corresponding implied volatilities. In contrast, while the option prices are always decreasing with the strike, this is not the case with the implied volatilities except for $\delta=2\%$ and $\delta=3\%$. In the case where $\delta=1\%$, the implied volatility is increasing with respect to the strike. Hence with the help of the parameter $\delta$, one can observe different levels of skewness.

\begin{table}[H]
\centering
\resizebox{12cm}{!}{
\begin{tabular}{|c|c|c|c|c|c|c|}
   \hline
     $k$ (bps) &\multicolumn{3}{c|}{Payer} & \multicolumn{3}{c|}{Implied vol (\%)} \\
   \hline
    &$\delta=1$\% &$\delta=2$\%& $\delta=3$\% & $\delta=1$\% & $\delta=2$\% & $\delta=3$\% \\
   \hline
   52.9 & 0.004655 &  0.007214& 0.009276  & 69.44& 133.84& 196.80\\
   \hline
    66.2& 0.003739 & 0.006077 &0.008032 & 73.36& 124.16& 173.70\\
   \hline
  79.4& 0.003107 & 0.005298 & 0.007081 & 76.85& 121.08  &161.57\\
   \hline
\end{tabular}
}
\caption{CDS options and corresponding Black volatilities (with spread $k^\star_0(a,b)=66.20$ pbs, $T_a=1$ and $T_b=3$) implied by the structural model using Monte Carlo simulation ($1.5\cdot 10^5$) paths and for various volatility parameter $\delta$ and different strikes (80\%, 100\%, 120\%)$k^\star_0(a,b)$ and $\sigma=$5\%.} 
\label{tab:CDSO}
\end{table}
Notice that in this example, we focus more on the numerical performances of the model and do not address the calibration problem. Hence, we use the model implied term structure given by the time-$0$ model survival probability curve as a CDS term structure. Calibration issues of the model to real market data will be investigated in a future work.

\section{Conclusion}
In this paper, a new structural model for credit risk has been proposed, generalizing earlier works. Our model deals with an incomplete information, where the default state and a noisy observation of the firm valued are accessible to the investor. It is therefore an extension of \cite{CocGemJea}, as the firm-value triggering the default is no longer restricted to be a continuous and invertible function of a Gaussian martingale, but can be any diffusion. 

This more general framework benefits however from a limited analytical tractability. Therefore, we propose a numerical method that relies on nonlinear filtering theory associated with recursive quantization. Compared to earlier works such as \cite{CalSag} or \cite{ProSag14}, our numerical procedure is based on the fast quantization method recently introduced in \cite{PagSagMQ}, which avoids the use of 
Monte Carlo simulations. A rigorous analysis of the global error induced to the estimation of the survival processes is performed. We analyze the shapes of the default probabilities which are characterized by a path-dependent feature keeping the memory of all the path of the observed process. Eventually we quantify the impact of the volatility of the noise impacting the firm-value process on the pricing of CDS options and the corresponding implied volatilities using a hybrid Monte Carlo-optimal quantization method.\medskip

In future research, we will first investigate the calibration issues of the model which can be tackled by either using observed prices or CDS quotes. In this case, our model can be easily extended to other works dealing with exact calibration to survival probabilities such as including a specific time-dependent barrier \cite{BrigTar04} or using time change techniques \cite{Mbaye2019}. Another possible research area is to deal with the price of general default sensitive securities. While we have derived a full quantization scheme to estimate the conditional default probabilities, this was not the case in the pricing of CDS option which required additional Monte Carlo simulations in order to be estimated. To derive a full quantization scheme for the pricing of defaultable claims, a possible route is to exploit the functional quantization method.

 \newpage
\bibliographystyle{plain}
\bibliography{NLfilteringbib}

\end{document}